\documentclass[12pt]{iopart}

\usepackage{iopams}
\usepackage{setstack}
\usepackage{amscd}

\usepackage[square,sort&compress,numbers]{natbib}
\usepackage[colorlinks=true,citecolor=blue,urlcolor=green]{hyperref}
\newcommand{\op}[1]{{\ensuremath{\sf #1}}}
\newcommand{\opt}[1]{{\ensuremath{\widetilde{\sf #1}}}}

\newtheorem{theorem}{Theorem}
\newtheorem{proposition}[theorem]{Proposition}

\newenvironment{proof}[1][Proof]{\begin{trivlist}
\item[\hskip \labelsep {\textit{#1.}}]}{\end{trivlist}}

\newcommand{\qed}{\hfill \ensuremath{\Box}}

\hypersetup{
    pdftitle={Noether symmetries and conservation laws in the classical problem of a particle with linear damping}, 
    pdfauthor={T. Gourieux and R. Leone},   
    pdfcreator={T. Gourieux and R. Leone},   
    pdfkeywords={Classical mechanics}{Noether symmetries}{Lie symmetries}{first integrals}{non-conservative systems}{Bateman-		Caldirola-Kanai Lagrangian}	
}

\begin{document}

\title[Classical Noether's theory with application to the linearly damped particle]{Classical Noether's theory with application to the linearly damped particle}

\author{Rapha\"{e}l Leone and Thierry Gourieux}

\address{Universit\'e de Lorraine, IJL, Groupe de Physique Statistique (UMR CNRS 7198)\\ 
F-54506 Vand\oe{}uvre-l\`es-Nancy cedex, France}

\ead{\href{mailto:thierry.gourieux@univ-lorraine.fr}{thierry.gourieux@univ-lorraine.fr}}

\begin{abstract}
This paper provides a modern presentation of Noether's theory in the realm of classical dynamics, with application to the problem of a particle submitted to both a potential and a linear dissipation. After a review of the close relationships between Noether symmetries and first integrals, we investigate the variational point symmetries of the Lagrangian introduced by Bateman, Caldirola and Kanai. This analysis leads to the determination of all the time-independent potentials allowing such symmetries, in the one-dimensional and the radial cases. Then we develop a symmetry-based transformation of Lagrangians into autonomous others, and apply it to our problem. To be complete, we enlarge the study to Lie point symmetries which we associate logically to Noether ones. Finally, we succinctly address the issue of a `weakened' Noether's theory, in connection with `on-flows' symmetries and non-local constant of motions, for it has a direct physical interpretation in our specific problem. Since the Lagrangian we use gives rise to simple calculations, we hope that this work will be of didactic interest to graduate students, and give teaching material as well as food for thought for physicists regarding Noether's theory and the recent developments around the idea of symmetry in classical mechanics.

\end{abstract}

\vspace{2pc}

\maketitle

\section{Introduction}

In Lagrangian mechanics, a primary method for taking into account non-conservative forces is to add their generalized counterparts in the Euler-Lagrange equations~\cite{Whittaker,Goldstein}. But by doing so, one loses the variational framework on which the `wonderful'~\cite{Neuenschwander} Noether's theory~\cite{Noether,Sarlet_Cantrijn} relies naturally. Then, if we are interested in the study of conservation laws, one might use alternative techniques, focused on the equations of motion. Among the most elegant are Lie's theory~\cite{SophusLie,Olver,IbragimovReview} and $\lambda$-symmetries~\cite{Muriel_Romero2001,Muriel_Romero2003}, Jacobi last multipliers~\cite{Jacobi1842,Jacobi1866,Whittaker,Nucci2008bis}, or methods based upon integrating factors~\cite{Duarte,Anco_Bluman,Cheb-Terrab}. They are highly interconnected~\cite{Nucci2005,Muriel_Romero2009,Mohanasubha,Nucci2013} and each of them stands more or less on some symmetry concept. A second, albeit related, approach consists in trying to find a `non-standard' Lagrangian encoding all the dynamics so as to recover a genuine variational framework on which Noether's theory can be applied. That demands to obtain a solution to the inverse Lagrange problem, as first posed by Helmholtz in 1887~\cite{Helmholtz,Whittaker,Havas}, e.g. by seeking a Jacobi last multiplier~\cite{Nucci2008}. However, solutions may be not tractable, or none may exist, or even the resulting Lagrangians may be very cumbersome. A last possibility is to extend Noether's theory to the primary formulation~\cite{Djukic_Vujanovic}, or at the level of d'Alembert's principle~\cite{Vujanovic}.

Fortunately, after the seminal works of Bateman~\cite{Bateman}, Caldirola~\cite{Caldirola}, and Kanai~\cite{Kanai} (BCK), the problem of a particle submitted to both a potential $V$ and a linear drag force $\bi f=-2m\gamma\bi v$, where $m$ is the particle's mass and $\gamma$ the dissipation rate, is known to admit the quite simple Lagrangian
\begin{equation}
L_{\mathrm{BCK}}(\bi r,\bi v,t)=\left(\frac12\,m\,\bi v^2-V\right)\rme^{2\gamma t}.\label{BCK}
\end{equation}
Indeed, the corresponding Euler-Lagrange equations are, in vector notation,
\begin{equation*}
\boldsymbol{\op E}(L_{\rm BCK})=\left(m\,\frac{\rmd\bi v}{\rmd t}+\frac{\partial V}{\partial\bi r}+2m\gamma\bi v\right)\rme^{2\gamma t}=\mathbf 0.
\end{equation*}
Since the weight factor $\rme^{2\gamma t}$ does never vanish during time, the above equations give a faithful account of the dynamics\footnote{The BCK Lagrangian can be adapted to a time-dependent dissipation rate; it suffices to replace $\gamma t$ by $\int\gamma(t)\,\rmd t$ in the exponential factor of~\eref{BCK}.}. This Lagrangian has interesting properties: (i) it is close in form and tends to the standard conservative Lagrangian $T-V$ when the dissipation rate $\gamma$ goes to zero; (ii) it is universally valid for all potentials. Those are among the reasons why, for decades, the BCK Lagrangian has attracted attention~\cite{Denman1966,Denman1972,Lemos1979,Bahar1981,Kobe1986,Rosales1989,Razavy2005} and why it has become a kind of paragon for the thorny question of canonical quantization of non-conservative systems~\cite{Kanai,Kerner1958,Senitzky1960,Dekker1981,Baldiotti2011}, first and foremost addressed within the ubiquitous harmonic oscillator. In this article, however, we will not enter this epistemological debate. We will remain at the classical standpoint: our purpose here is to analyse the variational symmetries of $L_{\rm BCK}$ and the conservation laws they induce.

The paper is organized as follows. We give in \sref{sec:Noetherthm} an overview of Noether's theorem in classical mechanics and highlight the intimate connections between Noether symmetries and first integrals. Then, we begin \sref{sec:NPS} by seeking Noether point symmetries (NPS) of the one-dimensional BCK Lagrangian. It brings us to all the time-independent potentials admitting at least one NPS. They are divisible into four classes, including the very special family of at most quadratic polynomials which are known to possess the maximal number of independent NPS: 5~\cite{Mahomed_Kara_Leach}. As a second step, we make an excursion into the central problems in three dimensions. On the one hand, rotational symmetry about the center of force emerges naturally, but, on the other hand, we show that the angular momentum breaks some of the symmetries previously identified in the unidimensional case. 

\Sref{sec:autonomous} is devoted to symmetry-based transformations of Lagrangians into autonomous others. By construction, the latter lead to Hamiltonians which coincide with the first integrals generated by the symmetries. More precisely, we develop a systematic scheme for carrying out such a mapping from the existence of a NPS, and apply it to the BCK Lagrangian. In \sref{sec:Lie}, we enlarge the study to Lie point symmetries (LPS) of the Euler-Lagrange equations. We show that the time translation invariance is the only additional point symmetry, with the notable exception of potentials at most quadratic. Indeed, they reach a LPS algebra of maximal dimension: 8~\cite{Mahomed_Kara_Leach}. It is \textit{en passant} the opportunity to give the eight first integrals associated to a basis of those algebra for the linear potential, which we did not find in the literature. Then, thanks to the converse of Noether's theorem, we associate to each of these additional symmetries a Noether one. Finally, we discuss in \sref{sec:action} a weak version of Noether's theorem, in terms of which any first integral amounts to a local expression of the BCK action functional along the solution curves.

\section{Noether's theorem in classical mechanics}\label{sec:Noetherthm}

\subsection{The general statement of Noether's theorem; Killing-type equations}

Let us consider a dynamical system governed by a regular Lagrangian $L$ in terms of $n$ generalized coordinates $q^i$. Hamilton's variational principle applied to the action functional
\begin{equation*}
A:=\int L(q,\dot q,t)\,\rmd t
\end{equation*}
yields a set of $n$ second-order differential equations $\op E_i(L)=0$, where
\begin{equation*}
\op E_i:=\frac{\partial}{\partial q^i}-\frac{\rmd}{\rmd t}\frac{\partial}{\partial\dot q^i}
\end{equation*}
is the Euler-Lagrange operator associated to $q^i$. Under the regularity assumption of $L$, one can isolate the accelerations $\ddot q^{\,i}$ to bring the equations to the normal form
\begin{equation*}
\ddot q^{\,i}=\Omega^i(q,\dot q,t).
\end{equation*}
Noether's theorem ensues from a variational principle involving the action functional as well. Here, the key element is its variations under infinitesimal transformations of the coordinates and time. The most general ones we will consider read formally
\begin{equation}
q^i\longrightarrow\tilde q^{\,i}=q^i+\varepsilon\,\xi^i(q,\dot q,t)\quad,\quad t\longrightarrow\tilde t=t+\varepsilon\,\tau(q,\dot q,t),\label{transformation}
\end{equation}
where $\varepsilon$ is the infinitesimal parameter, the functions $\tau$, $\xi^i$ being smooth with respect to their arguments. We are dealing with the very familiar point transformations when these functions do not depend on the velocities. Adopting Einstein's summation convention on repeated indices, let us introduce the generator of the transformation~\eref{transformation}
\begin{equation}
\op X:=\tau\,\frac{\partial}{\partial t}+\xi^i\,\frac{\partial}{\partial q^i}\,,\label{generateur}
\end{equation}
which allows to write its natural action on any differentiable function $G(q,t)$ as
\begin{equation*}
G(q,t)\longrightarrow G(\tilde q,\tilde t)=G(q,t)+\varepsilon\,\op X(G(q,t))+\Or(\varepsilon^2),
\end{equation*}
after a first order Taylor expansion. The transformation maps any curve $t\mapsto q(t)$ to a curve $\tilde t\mapsto\tilde q(\tilde t)$ and affects the velocities according to
\begin{equation*}
\dot q^i\longrightarrow\frac{\rmd \tilde q^{\,i}}{\rmd\tilde t}=\frac{\rmd \tilde q^{\,i}/\rmd t}{\rmd\tilde t/\rmd t}=\dot q^i+ \varepsilon(\dot\xi^i-\dot q^i\dot\tau)+\Or(\varepsilon^2).
\end{equation*}
By extension, the effect of~\eref{transformation} on velocity-dependent functions $G(q,\dot q,t)$ becomes
\begin{equation}
G(q,\dot q,t)\longrightarrow G\Big(\tilde q,\frac{\rmd\tilde q}{\rmd\tilde t},\tilde t\Big)=G(q,\dot q,t)+\varepsilon\,\op X^{[1]}(G(q,\dot q,t))+\Or(\varepsilon^2)\label{extension}
\end{equation}
where
\begin{equation*}
\op X^{[1]}:=\op X+(\dot\xi^i-\dot q^i\dot\tau)\frac{\partial}{\partial\dot q^i}
\end{equation*}
is the first prolongation of the generator. Successive prolongations $\op X^{[n]}$ can be deduced recursively to act on dynamical functions of $t$, $q$, $\dot q$, $\ddot q$, and so forth until the $n$-th time-derivative of $q$ (in \sref{sec:Lie} we will use the second prolongation). Now, the effect of the transformation~\eref{transformation} on the action functional is evaluated via the variation
\begin{equation}
\delta A=\int_{\tilde t_1}^{\tilde t_2}L\left(\tilde q,\frac{\rmd\tilde q}{\rmd\tilde t},\tilde t\right)\rmd\tilde t-\int_{t_1}^{t_2}L(q,\dot q,t)\,\rmd t
\end{equation}
for any path $t\mapsto q(t)$ in the configuration space, between two arbitrary instants $t_1$ and $t_2$. Then, it is straightforward to derive from~\eref{extension} the formula
\begin{equation}
\delta A=\varepsilon\int_{t_1}^{t_2}\Big(\op X^{[1]}(L)+\dot\tau L\Big)\,\rmd t+\Or(\varepsilon^2).\label{deltaA}
\end{equation}
We say that, under the transformation~\eref{transformation}, the functional is \emph{invariant up to a divergence term} $f$ if the integrand in~\eref{deltaA} is the total time derivative of some function $f(q,\dot q,t)$~\cite{Bessel-Hagen,Sarlet_Cantrijn}:
\begin{equation}
\op X^{[1]}(L)+\dot\tau L=\dot f.\label{RT1}
\end{equation}
Equation~\eref{RT1} is the so-called Rund-Trautman identity~\cite{Rund,Trautman}. After some algebra, it can be re-written
\begin{equation}
(\xi^i-\dot q^i\tau)\op E_i(L)+\frac{\rmd}{\rmd t}\left[L\tau+\frac{\partial L}{\partial\dot q^i}(\xi^i-\dot q^i\tau)\right]=\dot f\,.\label{RT2}
\end{equation}
Noether's theorem follows directly from~\eref{RT2}.

\begin{theorem}[Noether]
If the action functional is invariant under the
infinitesimal transformation~\eref{transformation}, up to the divergence term $f$, then the quantity
\begin{equation}
I(q,\dot q,t)=f-L\tau-\frac{\partial L}{\partial\dot q^i}(\xi^i-\dot q^i\tau)\label{FINoether}
\end{equation}
is a first integral of the problem.
\end{theorem}

Under these circumstances, we say that the transformation is a \emph{Noether symmetry} (or \emph{variational symmetry}) of the problem, and that it generates the first integral $I$. It is called \emph{strict} when $f=0$. Since a transformation is entirely characterized by its generator, and \textit{vice versa}, we will from now on identify these two notions and practically speak in terms of transformations or Noether symmetries $\op X$. The most familiar ones are the \emph{point symmetries}. For instance, any cyclic coordinate $q^i$ stands for the strict Noether point symmetry $\partial_{q^i}$ which generates the momentum $p_i=\partial_{\dot q^i}L$ as first integral. In the same way, $\partial_t$ is such a symmetry of autonomous Lagrangians from which follows the conservation of $\dot q^i\partial_{\dot q^i}L-L$ along the solution curves, viz. the Beltrami identity. 

Unfortunately, most of the symmetries cannot be brought to light by simply taking a look at the Lagrangian. The only way to unearth `hidden' symmetries is to seek solutions of the Rund-Trautman identity. We would like to insist upon the fact that~\eref{RT1} is assumed to hold for all the paths $t\mapsto q(t)$, not only along solution curves of the Euler-Lagrange equations (the actual paths): roughly speaking, we deal with `strong' solutions and not only `on-flow' ones~\cite{Gorni_Zampieri} (the nuance holds only for non-point transformations). Thus, to avoid any confusion inherent to the overdot notation, it would be preferable, at least for further theoretical considerations, to make use of the total time-derivative operator
\begin{equation}
\op D=\frac{\partial}{\partial t}+\dot q^i\frac{\partial}{\partial q^i}+\ddot q^{\,i}\frac{\partial}{\partial\dot q^i}+\dots\label{timederivative}
\end{equation}
and to identify $\dot G$ with $\op D(G)$ for any dynamical function $G$. Hence, equation~\eref{RT1} is definitely an identity in the variables $t$, $q^i$, $\dot q^i$, $\ddot q^{\,i}$. It is clearly linear in the $\ddot q^{\,i}$. Therefore, its coefficients must vanish separately, providing the $n+1$ `Killing-type' equations
\numparts
\begin{eqnarray}
\fl\frac{\partial\tau}{\partial\dot q^i}\left(L-\dot q^j\frac{\partial L}{\partial\dot q^j}\right)+\frac{\partial\xi^j}{\partial\dot q^i}\frac{\partial L}{\partial\dot q^j}=\frac{\partial f}{\partial\dot q^i}\;;\label{Killing1}\\
\fl\tau\frac{\partial L}{\partial t}+\xi^i\frac{\partial L}{\partial q^i}+\frac{\partial L}{\partial\dot q^i}\left(\dot q^j\frac{\partial\xi^i}{\partial q^j}+\frac{\partial\xi^i}{\partial t}\right)+\left(\dot q^i\frac{\partial\tau}{\partial q^i}+\frac{\partial\tau}{\partial t}\right)\left(L-\dot q^i\frac{\partial L}{\partial\dot q^i}\right)=\dot q^i\frac{\partial f}{\partial q^i}+\frac{\partial f}{\partial t}\,.\label{Killing2}
\end{eqnarray}
\endnumparts
By seeking the Noether point symmetries (NPS), the left-hand sides of the $n$  equations~\eref{Killing1} vanish, implying $f=f(q,t)$ and leaving us with the single equation~\eref{Killing2}. The latter is frequently solvable in a completely algorithmic way, as will be the case for $L_{\rm BCK}$ in \sref{sec:NPS}. 

Before entering further into the subject of Noether symmetries, let us end this paragraph with a comforting property ensuring their preservation under Lagrangian gauge transformations.

\begin{proposition}\label{lemme}
Let $\op X$ be a Noether symmetry of $L$, with divergence term $f$. Let $L\to \widehat L=L+\op D(\Lambda)$ be a gauge transformation induced by some function $\Lambda(q,t)$. Then, $\op X$ is a Noether symmetry of the new Lagrangian $\widehat L$, with divergence term $f+\op X(\Lambda)$, and generates the same first integral ($\widehat I=I$).
\end{proposition}

\begin{proof}
With $\op X$ given by~\eref{generateur}, it is a simple task to check the validity of
\begin{equation}
\op D(\op X(\Lambda))-\op X^{[1]}(\op D(\Lambda))=\op D(\tau)\op D(\Lambda),\label{commutation}
\end{equation}
for any function $\Lambda(q,t)$. Whence,
\begin{eqnarray*}
\op X^{[1]}(\widehat L)+\op D(\tau)\widehat L&=\op X^{[1]}(L)+\op D(\tau)L+\op X^{[1]}(\op D(\Lambda))+\op D(\tau)\op D(\Lambda)\\
&=\op D(f+\op X(\Lambda)).
\end{eqnarray*}
The conservation of the first integral is straightforwardly verified.\qed
\end{proof}

As an interesting corollary of the proposition, one sees that if the function $\Lambda$ satisfies the `gauge condition'
\begin{equation}
f+\op X(\Lambda)=0\label{gauge}
\end{equation}
then $\op X$ is a strict Noether symmetry of the new Lagrangian $\widehat L$.

\subsection{Beyond the general statement}

The set of transformations~\eref{generateur} carries a natural structure of Lie algebra with respect to the bracket $[\op X,\op Y]=\op X\op Y-\op Y\op X$ between vector fields. It can be proved~\cite{Sarlet_Cantrijn} that the subset of Noether's symmetries of $L$ forms a subalgebra. In particular, it has a vector space structure: if $\op X$ and $\op X'$ are two Noether symmetries, and if $\lambda$ is a scalar, then $\op X+\op X'$ and $\lambda\,\op X$ are also Noether symmetries, as it is clear from~\eref{FINoether}. With obvious notations, the generated first integrals are respectively $I+I'$ and $\lambda\, I$. For the remainder of this article, transformations differing by a nonzero multiplicative constant $\lambda$ will not be distinguished since first integrals $I$ and $\lambda I$ are essentially the same.

Let us now give an overview of the most significant facts regarding the interplay between first integrals and Noether symmetries. By saying that $I$ is a first integral of the problem we mean that it is a function of $(q,\dot q,t)$ which is a constant of the motion, i.e.
\begin{equation}
\op D(I)=0\quad{\rm when}\quad\{\op E_i(L)=0\}.\label{FI}
\end{equation}
Introducing the vector field associated to the dynamics,
\begin{equation*}
\op\Gamma=\frac{\partial}{\partial t}+\dot q^i\,\frac{\partial}{\partial q^i}+\Omega^i\,\frac{\partial}{\partial\dot q^i}\,,
\end{equation*}
the condition~\eref{FI} may be brought to the more compact form
\begin{equation*}
\op\Gamma(I)=0.
\end{equation*}
The total time-derivative of $I$ is thus
\begin{equation*}
\op D(I)=\frac{\partial I}{\partial t}+\dot q^i\,\frac{\partial I}{\partial q^i}+\ddot q^{\,i}\,\frac{\partial I}{\partial\dot q^i}=\frac{\partial I}{\partial\dot q^i}(\ddot q^{\,i}-\Omega^i)=-\frac{\partial I}{\partial\dot q^i}\,g^{ij}\,\op E_j(L),
\end{equation*}
where $(g^{ij})$ is the inverse of $(\partial^2_{\dot q^{i}\dot q^j}L)$,  the Hessian matrix of $L$ with respect to the velocities. Hence, it is clear that a function $I(q,\dot q,t)$ is a first integral if and only if there exist $n$ other functions $\mu^i=\mu^i(q,\dot q,t)$ such that~\cite{Anco_Bluman}
\begin{equation}
\op D(I)=\mu^i\,\op E_i(L).\label{integratingfactors}
\end{equation}
They are the \emph{integrating factors} of the Euler-Lagrange equations going hand in hand with $I$. Now, assuming $I$ generated by a Noether symmetry~\eref{generateur}, one has by using~\eref{RT2} and~\eref{FINoether}
\begin{equation}
\op D(I)=(\xi^i-\dot q^i\tau)\op E_i(L).\label{integratingfactorNS}
\end{equation}
Alternatively stated, the integrating factors are precisely the quantities $\xi^i-\dot q^i\tau$ named \emph{characteristics} of the transformation. Conversely, let $I$ be a first integral verifying~\eref{integratingfactors}. Then, thanks to~\eref{RT2}, for any transformation $\op X$ such that $\xi^i-\dot q^i\tau=\mu^i$, one has
\begin{equation*}
\op X^{[1]}(L)+\dot\tau L=\op D\left[I+L\tau+\frac{\partial L}{\partial\dot q^i}(\xi^i-\dot q^i\tau)\right].
\end{equation*}
Noether's theorem admits thereby a converse which follows readily.

\begin{theorem}[Converse of Noether's theorem]
Let $I$ be a first integral and $\mu^i$ its integrating factors. Any transformation~\eref{generateur} having the $\mu^i$ as characteristics is a Noether symmetry of $L$, with divergence term
\begin{equation*}
f=I+L\tau+\frac{\partial L}{\partial\dot q^i}(\xi^i-\dot q^i\tau).
\end{equation*}
These (infinitely many) symmetries generate $I$. Furthermore, one has
\begin{equation}
\xi^i-\dot q^i\tau=\mu^i=-g^{ij}\frac{\partial I}{\partial\dot q^j}\,.\label{Hij}
\end{equation}
\end{theorem}

By virtue of the last theorem, a Noether symmetry is more properly an equivalence class, the underlying relation being the equality between characteristics. Actually, choosing a representative amounts to fixing a function $\tau$, the peculiar choice $\tau=0$ providing the so-called \emph{evolutionary} representative which has straightforward prolongations. Moreover, there is a one-to-one correspondence between the Noether symmetry classes and the set of first integrals~\cite{Sarlet_Cantrijn}. 

Finally, let us mention a notable property which enforces the relationship between a Noether symmetry $\op X$ and its associated first integral $I$: it can be proved that $I$ is itself a (first order) invariant of $\op X$, i.e.
\begin{equation*}
\op X^{[1]}(I)=0.
\end{equation*} 
We point out the fact that the `strong' nature of solutions of the Rund-Trautman identity is a key assumption of this deep property. As far as we know, it was first stated by Sarlet and Cantrijn in all its generality~\cite{Sarlet_Cantrijn}. 

\section{Noether Point Symmetries of the BCK Lagrangian}\label{sec:NPS}

\subsection{The unidimensional case}

Let us seek the NPS of the Lagrangian $L_{\rm BCK}$, in the case of a single coordinate $q$ and a time-independent potential $V(q)$ acting on a particle of unit mass ($m:=1$). Inserting $L_{\rm BCK}$ in~\eref{Killing2} provides an identity which is a cubic polynomial in $\dot q$ since $L_{\rm BCK}$ is itself quadratic with respect to that variable. Each of its four coefficients must vanish, yielding the determining system
\numparts
\begin{eqnarray}
\frac{\partial\tau}{\partial q}=0\;;\label{sys1}\\
-\frac12\frac{\partial\tau}{\partial t}+\gamma\tau+\frac{\partial\xi}{\partial q}=0\;;\label{sys2}\\
V\frac{\partial\tau}{\partial q}-\frac{\partial\xi}{\partial t}+\rme^{-2\gamma t}\frac{\partial f}{\partial q}=0\;;\label{sys3}\\
V\frac{\partial\tau}{\partial t}+\xi V'+2\gamma V\tau+\rme^{-2\gamma t}\frac{\partial f}{\partial t}=0,\label{sys4}
\end{eqnarray}
\endnumparts
where $V'$ stands for $\partial_qV$. Clearly, \eref{sys1}, \eref{sys2}, and \eref{sys3}, taken in this order, impose to $\tau$, $\xi$, and $f$ the following forms, regardless of the actual potential:
\begin{eqnarray}
\tau=\tau(t)\;;\label{tau}\\
\xi=\frac12(\dot\tau-2\gamma\tau) q+\psi(t)\;;\label{xi}\\
f=\left(\frac14(\ddot\tau-2\gamma\dot\tau)q^2+\dot\psi q\right)\rme^{2\gamma t}+\chi(t),\label{f}
\end{eqnarray}
where $\psi$ and $\chi$ are so far arbitrary functions of time, as well as $\tau$. These three functions characterize entirely the sought NPS including its divergence term. Inserting their above expressions in~\eref{sys4} provides a compatibility equation between them and the potential:
\begin{equation}
\fl\left(\frac{\dot\tau}{2}-\gamma\tau\right)qV'+\psi V'+(\dot\tau+2\gamma\tau)V+\left(\frac{\dddot\tau}{4}-\gamma^2\dot\tau\right)q^2+(\ddot\psi+2\gamma\dot\psi)q+\dot\chi \rme^{-2\gamma t}=0.\label{EdC}
\end{equation}
If $\gamma$ were zero, the triple $(\tau,\psi,\chi)=(1,0,0)$ would be an evident solution of~\eref{EdC} for any $V$, corresponding to the expected strict symmetry $\op X=\partial_t$. In the presence of dissipation, however, there is not such an universal solution. Indeed, suppose the family $\mathcal F:=\{qV',V',V,q^2,q,1\}$ to be linearly independent in the space of functions of the $q$ variable. In this scenario, the only way to satisfy~\eref{EdC} is to make each of their coefficients equal to zero, but it leads to $\tau=\psi=0$ i.e. to $\op X=\op 0$. Hence, the linear dependence of the family $\mathcal F$ is a necessary condition for the existence of a NPS.

\subsubsection{The polynomial potentials of degree less than or equal to 2.}

\begin{table}\footnotesize
\centering
\renewcommand{\arraystretch}{1.5}
\begin{tabular}{lll}
\br
Noether point symmetry & First integral & Divergence term\\
\mr
$\op X_1=\partial_q$ & $I_1=\frac{1}{2\gamma}\big(2\gamma\dot q-F\big)\rme^{2\gamma t}$ & $f_1=\frac{F}{2\gamma}\,\rme^{2\gamma t}$\\
$\op X_2=\rme^{-2\gamma t}\partial_q$ & $I_2=Ft-2\gamma q-\dot q$ & $f_2=Ft-2\gamma q$\\
$\op X_3=\rme^{2\gamma t}\big(2\gamma\partial_t+F\partial_q\big)$ & $I_3=\gamma I_1^2$ & $f_3=\frac{F}{4\gamma}\big(8\gamma^2q+F\big)\rme^{4\gamma t}$\\
$\op X_4=\rme^{-2\gamma t}\big(\partial_t+(Ft-2\gamma q)\partial_q\big)$ & $I_4=\frac12\,I_2^2$ & $f_4=2\gamma^2q^2+Fq(1-2\gamma t)+\frac12 F^2t^2$\\
$\op X_5=2\partial_t+(Ft-2\gamma q)\partial_q$ & $I_5=I_1\big(I_2-\frac{F}{2\gamma}\big)$ & $f_5=\frac{F}{4\gamma^2}\big(F(2\gamma t-1)+4\gamma^2q\big)\rme^{2\gamma t}$\\
\br
\end{tabular}
\caption{The five independent Noether point symmetries for the linear potential $V(q)=-Fq$.}\label{tableNPSlineaire} 
\end{table}

They constitute a highly specific subclass because they allow to reduce the family to only $\{q^2,q,1\}$. Translating the zero of potential\footnote{According to proposition~\ref{lemme}, a translation $V(q)\to V(q)+V_0$ in $L_{\rm BCK}$ amounts to a gauge transformation induced by $\Lambda(t)=\frac{1}{2\gamma}\,V_0\,\rme^{2\gamma t}$. Thus, as expected from obvious physical considerations, Noether symmetries are preserved under this unimportant operation.} or the origin of $q$ if necessary, they can be put under one of the two generic forms $V=Aq^2$ ($A$ arbitrary) or $V=-Fq$ ($F\ne0$). The former includes the problems of the harmonic oscillator $(A>0)$ as well as the `free' particle ($A=0$), while the latter refers to the problem of the particle submitted to a constant force $F$. 

Inserting the linear potential $-Fq$ in~\eref{EdC} provides the identity
\begin{equation*}
\left(\frac{\dddot\tau}{4}-\gamma^2\dot\tau\right)q^2+\left(\ddot\psi+2\gamma\dot\psi-\frac32\,\dot\tau F-\gamma\tau F\right)q+\dot\chi\,\rme^{-2\gamma t}=0.
\end{equation*}  
One can set $\chi=0$ and the general solution in terms of $\tau$ and $\psi$ is easily found. It is a linear combination of the five independent solutions summarized in table~\ref{tableNPSlineaire}, spanning thereby a five-dimensional NPS algebra. Only two of the resulting first integrals are functionally independent. We notice that the symmetry $\op X_5-\op X_4$ tends to $\partial_t$ as $\gamma$ goes to zero and generates a first integral $I=I_5-I_4+(F/2\gamma)^2$ reducing to the energy $E$ in the same limit. We call such a transformation an `energy-like' Noether symmetry. 

Applying the same procedure to $V=Aq^2$, one has
\begin{equation*}
\left(\frac{\dddot\tau}{4}+(A-\gamma^2)\dot\tau\right)q^2+\left(\ddot\psi+2\gamma\dot\psi+2A\psi\right)q+\dot\chi\,\rme^{-2\gamma t}=0.
\end{equation*} 
Here, $\psi$ and $\tau$ are uncoupled and the general solution is a linear combination of three solutions with $\psi=0$ and two others with $\tau=0$. Their form depends on the signum of  $\zeta:=A-\gamma^2$; that is, if $A>0$, on the oscillator's behaviour: underdamped ($\zeta>0$), critically damped ($\zeta=0$) or overdamped ($\zeta<0$). They can be found in references~\cite{Cervero} and~\cite{Choudhuri}.

\subsubsection{The other potentials.} The linear dependence of the family $\mathcal F$ imposes a linear combination
\begin{equation*}
\lambda_5qV'+\lambda_4V'+\lambda_3V+\lambda_2q^2+\lambda_1q+\lambda_0=0,
\end{equation*}
where the $\lambda_i$ are scalars, $\lambda_3$, $\lambda_4$, and $\lambda_5$ being non-simultaneously zero. It is actually an equation in $V$ whose general solution has a form depending on whether some of these three scalars are zero or not. It suffices to proceed exhaustively, excluding the option $\lambda_4=\lambda_5=0$ since it brings us back to the preceding case. Assuming a dimensionless coordinate $q$ if necessary, one obtains three possibilities:
\begin{equation*}
V(q)=A\,g(q)+Bq^2+Cq\quad{\rm with}\quad g(q)=\log(q),\,q^\alpha\;{\rm or}\;\rme^q,
\end{equation*}
where the constants $A$, $B$, $C$ are given by the scalars $\lambda_i$, and $\alpha\in\mathbb R-\{0,1,2\}$. Inserting, for example, $V(q)=q^\alpha$ in~\eref{EdC} gives a linear combination of $\{q^\alpha,q^{\alpha-1},q^2,q,1\}$ whose coefficients are functions of time. From the requirement that each of them must be zero, we deduce a unique NPS under some restrictions on $B$, $C$ and $\alpha$. We apply the same procedure to the two leftover potentials and summarize the results in table~\ref{tablepotentiels}. One can see that, in each case, the NPS is energy-like.

\begin{table}\footnotesize
\centering
\renewcommand{\arraystretch}{1.5}
\begin{tabular}{c|ccc}
\br
 & $V_1(q)=A\log(q)$ & $V_2(q)=Aq^{\alpha}+\frac{4\gamma^2\alpha}{(\alpha+2)^2}q^2$ & $V_3(q)=A\rme^q+8\gamma^2q$ \\
\mr
$\op X$ & $\rme^{-2\gamma t}\big(\partial_t-2\gamma q\partial_q\big)$ & $\rme^{2\gamma\frac{\alpha-2}{\alpha+2}t}\big(\partial_t-\frac{4\gamma }{\alpha+2}q\partial_q\big)$ & $\rme^{2\gamma t}\big(\partial_t-4\gamma\partial_q\big)$ \\
$f$ & $2\gamma(\gamma q^2+At)$ & $\frac{4\gamma^2(2-\alpha)}{(\alpha+2)^2}q^2\rme^{\frac{4\gamma\alpha}{\alpha+2}t}$ & $8\gamma^2(1-q)\rme^{4\gamma t}$ \\
$I$ & $\frac12(\dot q+2\gamma q)^2+A\log(q)+2\gamma At$ & $\big(\frac12(\dot q+\frac{4\gamma}{\alpha+2}q)^2+Aq^\alpha\big)\rme^{\frac{4\gamma\alpha}{\alpha+2}t}$ & $\big(\frac12(\dot q+4\gamma)^2+A\rme^q\big)\rme^{4\gamma t}$ \\ 
\br
\end{tabular}
\caption{The three additional potentials admitting one (and only one) NPS. In $V_1$ and $V_3$ the coordinate $q$ is assumed dimensionless. In $V_2$ the coefficient $\alpha$ belongs to $\mathbb R-\{-2,0,1,2\}$.}\label{tablepotentiels} 
\end{table}

To the best of our knowledge, $V_1$ and $V_3$ had not yet been identified. As for $V_2$, it was first revealed by Djukic and Vujanovic~\cite{Djukic_Vujanovic} in their work about extension of Noether's theorem to holonomic non-conservative dynamical systems. We emphasize that $V_2$ and $V_3$ depend on the dissipation rate $\gamma$. In some sense it means that the conservative force acting upon the particle is influenced by the resistive medium. At first glance it seems physically surprising but one should conceive an interaction with the environment having both a conservative part, deriving from $V$, and a non-conservative one. In fact, one has to keep in mind that, strictly speaking, we were not motivated by `pure' physical considerations, but rather by our desire to obtain Noether point symmetries.  

\subsection{The three-dimensional central case}

A direct extension of the preceding study is the problem of the damped motion in a central potential. The Lagrangian is now
\begin{equation*}
L_{\rm BCK}=\left(\frac12\,\boldsymbol{\dot r}^2-V(r)\right)\rme^{2\gamma t},
\end{equation*}
where $(x,y,z)$ are the cartesian coordinates of the particle and $r=\sqrt{x^2+y^2+z^2}$ its distance to the center of force (the origin). Using the vector notation, transformations read
\begin{equation*}
\op X=\tau\,\frac{\partial}{\partial t}+\boldsymbol\xi\cdot\frac{\partial}{\partial\bi r}\;\;,\qquad\boldsymbol\xi:=\xi_x\,\boldsymbol{\hat x}+\xi_y\,\boldsymbol{\hat y}+\xi_z\,\boldsymbol{\hat z}.
\end{equation*}
The Killing-type equation~\eref{Killing2} thereby obtained is yet again a cubic polynomial in velocities. The cubic, quadratic and linear monomials, taken one by one, provide the following forms for $\tau$, $\boldsymbol\xi$ and $f$:
\begin{eqnarray*}
\tau&=\tau(t)\;;\\
\boldsymbol\xi&=\frac12(\dot\tau-2\gamma\tau)\,\bi r+\boldsymbol\alpha\times\bi r+\boldsymbol\psi(t)\;;\\
f&=\left(\frac14(\ddot\tau-2\gamma\dot\tau)\,r^2+\boldsymbol{\dot\psi}\cdot\bi r\right)\rme^{2\gamma t}+\chi(t),
\end{eqnarray*}
where $\boldsymbol\alpha$ is a constant vector and $\boldsymbol\psi$ a time-dependent one. Then, the monomial of degree zero imposes the compatibility condition
\begin{equation}
\fl\left(\frac{\dot\tau}{2}-\gamma\tau\right)rV'+(\dot\tau+2\gamma\tau)V+\left(\frac{\dddot\tau}{4}-\gamma^2\dot\tau\right)r^2+\dot\chi\,\rme^{-2\gamma t}=-\left(\boldsymbol{\ddot\psi}+2\gamma\,\boldsymbol{\dot\psi}+\frac{V'}{r}\,\boldsymbol\psi\right)\cdot\bi r.\label{EdCcentral}
\end{equation}
Vector $\boldsymbol\alpha$ does not appear in~\eref{EdCcentral}. Consequently, whatever $\boldsymbol\alpha$ be, the transformation
\begin{equation}
\op X=(\boldsymbol\alpha\times\bi r)\cdot\frac{\partial}{\partial\bi r}\label{rotation}
\end{equation}
is a strict NPS of the Lagrangian, irrespective of the central potential $V$. It is the generator of the rotation of magnitude $\|\boldsymbol\alpha\|$ around the direction-vector $\boldsymbol{\hat \alpha}$. In other words, $L_{\rm BCK}$ is rotationally invariant, as expected by the isotropic nature of both the conservative force and the dissipation. The transformation~\eref{rotation} generates the first integral
\begin{equation*}
I(\boldsymbol\alpha)=(\boldsymbol\alpha\times\bi r)\cdot\frac{\partial L_{\rm BCK}}{\partial\boldsymbol{\dot r}}=\boldsymbol\alpha\cdot(\bi r\times\boldsymbol{\dot r})\,\rme^{2\gamma t}.
\end{equation*}
Since it is true for any $\boldsymbol\alpha$, one obtains the conservation of
\begin{equation*}
\boldsymbol\ell_0=(\bi r\times\boldsymbol{\dot r})\,\rme^{2\gamma t}=\boldsymbol\ell\,\rme^{2\gamma t},
\end{equation*}
where $\boldsymbol\ell$ is the angular momentum about the origin, equals to $\boldsymbol\ell_0$ at $t=0$. As a corollary, the motion takes place in the plane (through the origin) orthogonal to the constant vector $\boldsymbol\ell_0$, with an areal velocity decreasing exponentially with time.

Let us seek eventual supplementary symmetries. Passing to the spherical coordinates $(r,\vartheta,\varphi)$, such that
\begin{equation*}
x=r\sin\vartheta\cos\varphi\quad,\quad y=r\sin\vartheta\sin\varphi\quad,\quad z=r\cos\vartheta,
\end{equation*}
the left-hand side of~\eref{EdCcentral} has a spatial dependence contained only in the variable $r$, while the other side depends also on the two angles. Differentiating that equation two times with respect to $\vartheta$, the LHS disappears whereas the RHS changes sign. Hence, the latter must be zero as well as the former and~\eref{EdCcentral} splits into two:
\begin{eqnarray}
\left(\frac{\dot\tau}{2}-\gamma\tau\right)rV'+(\dot\tau+2\gamma\tau)V+\left(\frac{\dddot\tau}{4}-\gamma^2\dot\tau\right)r^2+\dot\chi\,\rme^{-2\gamma t}=0\;;\label{EdC1}\\
\boldsymbol{\ddot\psi}+2\gamma\,\boldsymbol{\dot\psi}+\frac{V'}{r}\,\boldsymbol\psi=\boldsymbol 0.\label{EdC2}
\end{eqnarray}
Equation~\eref{EdC2} can only be fulfilled by quadratic (or zero) potentials $V=A\,r^2$. Actually, these potentials are very special since they make $L_{\rm BCK}$ separable into three replications of the same one-dimensional Lagrangian, each one in terms of a single cartesian coordinate. Hence, these peculiar cases bring us back to the previous problem and contain $3\times5=15$ NPS in all. 

The non-quadratic potentials, however, must verify equation~\eref{EdC1} which is quite similar to~\eref{EdC} but differs by the absence of the function $\psi$. Consequently, the solutions of~\eref{EdCcentral} are exactly the ones of~\eref{EdC} having a zero $\psi$ (after replacing $q$ by $r$). A quick inspection of tables~\ref{tableNPSlineaire} and~\ref{tablepotentiels} shows that $V_1$ and $V_2$ have an additional symmetry, in contrary to $V_3$ and the linear potential.

\section{From Noether point symmetries to equivalent autonomous problems}\label{sec:autonomous}

As a rule in physics, a symmetry provides suitable coordinates for the description of a problem in a simpler way. Having this idea in mind, we will show how such coordinates arise from a NPS, and secondly how they can be used to map a general Lagrangian problem into an autonomous one, whose Hamiltonian is precisely the first integral $I$ generated by the symmetry. Then, we will apply the procedure to $L_{\rm{BCK}}$. 

Let $L$ be any Lagrangian in terms of $n$ coordinates $q^i$, and suppose that it possesses the NPS
\begin{equation*}
\op X=\tau(q,t)\frac{\partial}{\partial t}+\xi^i(q,t)\frac{\partial}{\partial q^i}\,,
\end{equation*}
with divergence term $f$. Since we will perform changes of variables, it seems preferable to use more intrinsic formulations. First, noticing that $\op X(t)=\tau$ and $\op X(q^i)=\xi^i$, the NPS may be written
\begin{equation*}
\op X=\op X(t)\frac{\partial}{\partial t}+\op X(q^i)\frac{\partial}{\partial q^i}\,.
\end{equation*}
Then, using the total time-derivative operator~\eref{timederivative}, the Rund-Trautman identity~\eref{RT1} verified by $\op X$ becomes
\begin{equation}
\op X^{[1]}(L)+\op D(\op X(t))L=\op D(f).\label{intrinsicRT}
\end{equation}
Now, let us consider an invertible change of variables $(q,t)\to(Q,T)$, where $Q$ is seen as the new set of coordinates and $T$ the new time. It provides an alternative expression to the vector field $\op X$:
\begin{equation*}
\op X=\op X(T)\frac{\partial}{\partial T}+\op X(Q^i)\frac{\partial}{\partial Q^i}\,.
\end{equation*}
To avoid any confusion, we denote the total derivative with respect to $T$ in this new representation by the `prime' symbol, the corresponding operator being
\begin{equation*}
\opt D=\partial_T+{Q^i}'\,\frac{\partial}{\partial Q^i}+\dots=\opt D(t)\,\op D.
\end{equation*} 
The dynamics in the $(Q,T)$ variables may now be described by the transformed Lagrangian
\begin{equation*}
\widetilde L(Q,Q',T)=L\,\opt D(t).
\end{equation*}
Here, the identity~\eref{commutation} is equivalent to 
\begin{equation}
\opt D(\op X(G))-\op X^{[1]}(\opt D(G))=\opt D(\op X(T))\opt D(G),\label{commutation2}
\end{equation}
for any function $G$ defined over the extended configuration space. Using~\eref{commutation2} with $G\equiv t$, one readily derives from~\eref{intrinsicRT}
\begin{equation*}
\op X^{[1]}(\widetilde L)+\opt D(\op X(T))\widetilde L=\opt D(t)\Big[\op X^{[1]}(L)+\op D(\op X(t))L\Big]=\opt D(t)\op D(f)=\opt D(f).
\end{equation*}
In words, $\op X$ remains a NPS of $\widetilde L$, with the same divergence term. The corresponding first integral, as logic suggests, is still $I$ (see the appendix):
\begin{equation}
\widetilde I=f-\widetilde L\,\op X(T)-\frac{\partial\widetilde L}{\partial {Q^i}'}\Big(\op X(Q^i)-{Q^i}'\,\op X(T)\Big)=I.\label{appendice}
\end{equation}
The variables $(Q,T)$ are so far arbitrary but, in order to take into account the symmetry in the simplest possible way, one has to choose proper ones. They are the \emph{canonical variables} of $\op X$, such that $\op X(Q^i)=0$ and $\op X(T)=1$. They reduce $\op X$ to the translation symmetry $\partial _T$. Performing a gauge change if necessary, thanks to proposition~\ref{lemme} and formula~\eref{gauge}, one has thereby obtained a Lagrangian which is $T$-independent. All this can be summarized in the following theorem.

\begin{theorem}\label{thm:transformation}
Let $(Q,T)$ be canonical variables of $\op X$, such that $\op X(Q^i)=0$ and $\op X(T)=1$, and let $\Lambda$ be a function satisfying the gauge condition $f+\op X(\Lambda)=0$. Then, in the $(Q,T)$ variables, the new Lagrangian $\widetilde L=L\,\opt D(t)+\opt D(\Lambda)$ is explicitly $T$-independent: $\widetilde L=\widetilde L(Q,Q')$. Furthermore, the induced Hamiltonian is precisely the first integral $I$:
\begin{equation*}
\widetilde H=P_i{Q^i}'-\widetilde L=I\quad{\it with}\quad P_i=\frac{\partial\widetilde L}{\partial{Q^i}'}\,.
\end{equation*}
\end{theorem}

Let us apply the theorem to $L_{\rm BCK}$ in one dimension for which we assume a NPS with a nonzero $\tau$ function. Taking advantage of the $q$-independence of $\tau$, a simple time rescaling $t\to T$ is possible, through
\begin{equation*}
T=\int\frac{\rmd t}{\tau}\,.
\end{equation*}
To obtain the canonical coordinate $Q$, we seek an invariant of the differential equation
\begin{equation*}
\frac{\rmd t}{\tau}=\frac{\rmd q}{\xi}\,,
\end{equation*}
with $\xi$ given by~\eref{xi}. One easily finds
\begin{equation*}
Q=\frac{\rme^{\gamma t}}{\sqrt\tau}\,q-\int\frac{\rme^{\gamma t}\psi}{\tau^{3/2}}\,\rmd t.
\end{equation*}
The new velocity is thus given by
\begin{equation}
Q'=\frac{\op D(Q)}{\op D(T)}=\tau\left(\frac{\partial Q}{\partial t}+\dot q\,\frac{\partial Q}{\partial q}\right)=\rme^{\gamma t}\sqrt\tau\left(\dot q-\frac{\xi}{\tau}\right).\label{Q'}
\end{equation}
Hereafter $Q$ and $T$ are the independent variables. Since $f$ is quadratic in $q$ which is itself linear in $Q$, the function $f$ has the form
\begin{equation*}
f=f_2(T)Q^2+f_1(T)Q+f_0(T).
\end{equation*}
Expanding~\eref{f} gives the precise expressions of $f_0$, $f_1$, and $f_2$. In particular, the last two ones are
\begin{eqnarray*}
f_1&=\frac{1}{2}(\ddot\tau-2\gamma\dot\tau)\,\tau\int\frac{\rme^{\gamma t}\psi}{\tau^{3/2}}\,\rmd t+\dot\psi\,\rme^{\gamma t}\sqrt\tau\;;\\
f_2&=\frac{1}{4}(\ddot\tau-2\gamma\dot\tau)\,\tau.
\end{eqnarray*}
The gauge condition $f+\op X(\Lambda)=f+\partial_T\Lambda=0$ is automatically fulfilled if one sets
\begin{equation*}
\Lambda=-Q^2\int f_2(T)\,\rmd T-Q\int f_1(T)\,\rmd T-\int f_0(T)\,\rmd T.
\end{equation*}
However, the construction of the explicitly $T$-independent Lagrangian needs only the knowledge of $\opt D(\Lambda)$, that is, the partial derivatives of $\Lambda$. The first one, $\partial_T\Lambda=-f$, is already known and the integration of $f_1$ and $f_2$ yields
\begin{equation*}
\frac{\partial\Lambda}{\partial Q}=-2\,Q\int f_2(T)\,\rmd T-\int f_1(T)\,\rmd T=-\frac{\xi}{\sqrt\tau}\,\rme^{\gamma t}.
\end{equation*}
Then, using~\eref{Q'}, the autonomous Lagrangian is
\begin{equation}
\widetilde L(Q,Q')=L\,\opt D(t)+\opt D(\Lambda)=\frac{1}{2}\, Q'^2-\widetilde V(Q),\label{AutonomousLagrangian}
\end{equation}
where
\begin{equation*}
\widetilde V=\left(V\,\tau-\frac{\xi^2}{2\tau}\right)\rme^{2\gamma t}+f
\end{equation*}
appears as the new potential. For consistency, one can check the $T$-independence of $\widetilde V$:
\begin{equation*}
\frac{\partial\widetilde V}{\partial T}=\op X(\widetilde V)=\tau\frac{\partial\widetilde V}{\partial t}+\xi\frac{\partial\widetilde V}{\partial q}=0,
\end{equation*}
by virtue of equations~\eref{sys1}-\eref{sys4}. As a last step, we derive from~\eref{AutonomousLagrangian} the momentum $P=Q'$ and the conserved Hamiltonian
\begin{equation*}
\widetilde H(Q,P)=\frac12 P^2+\widetilde V(Q)=I.
\end{equation*}
In this new picture, $I$ is nothing but the energy of the equivalent autonomous problem. We summarize on table~\ref{tableH} the explicit transformations together with their associated Hamiltonian for the three potentials $V_1$, $V_2$, and $V_3$ found above. We also include the case of the linear potential $-Fq$ whose three symmetries $\op X_3$, $\op X_4$, $\op X_5$ have a nonzero $\tau$ (see table~\ref{tableNPSlineaire}). Finally, applying also theorem~\ref{thm:transformation} to the NPS whose $\tau$ is zero, we end table~\ref{tableH} with the two leftover NPS of the linear potential. The same procedure may be followed in the three-dimensional central case as well.

\begin{table}\small
\centering
\renewcommand{\arraystretch}{1.5}
\begin{tabular}{llll}
\br
$V$ & $T$ & $Q$ & $\widetilde H$\\
\mr
$V_1(q)=A\log(q)$ & $\frac{1}{2\gamma}\,\rme^{2\gamma t}$ & $q\,\rme^{2\gamma t}$ & $\frac12 P^2+A\log(Q)$\\ 
$V_2(q)=Aq^\alpha+\frac{4\gamma^2\alpha}{(\alpha+2)^2}\,q^2$ & $\frac{1}{2\gamma}\frac{2+\alpha}{2-\alpha}\,\rme^{2\gamma\frac{2-\alpha}{2+\alpha}t}$ & $q\,\rme^{\frac{4\gamma t}{\alpha+2}}$ &  $\frac12 P^2+AQ^\alpha$\\
$V_3(q)=A\rme^{q}+8\gamma^2q$ & $-\frac{1}{2\gamma}\,\rme^{-2\gamma t}$ & $q+4\gamma t$ & $\frac12 P^2+A\rme^{Q}$\\
\hfill($\op X_3$) & $\rme^{-2\gamma t}$ & $2\gamma q-Ft$ & $\frac12\,P^2$\\
\hfill($\op X_4$) & $\rme^{2\gamma t}$ & $(F(1-2\gamma t)+4\gamma^2q)\,\rme^{2\gamma t}$ & $\frac12 P^2$\\
 $V(q)=-Fq$ \hfill($\op X_5$) & $\gamma t$ & $(F(1-\gamma t)+2\gamma^2q)\,\rme^{\gamma t}$ & $\frac12 P^2-\frac12 Q^2$\\
\hfill($\op X_1$) & $2\gamma^2q$ & $\rme^{2\gamma t}$  & $(2\sqrt P+F)Q$\\
\hfill($\op X_2$) & $\frac{\gamma^2}{2}\,q\rme^{2\gamma t}$ & $\gamma t$ & $2\rme^{-Q}\sqrt P+FQ$\\
\br
\end{tabular}
\caption{Autonomous one-dimensional Hamiltonians deduced from the symmetries.}\label{tableH}
\end{table}

\section{Lie point symmetries}\label{sec:Lie}

A Noether symmetry of a Lagrangian $L$ has the well-known property of permuting solutions of the Euler-Lagrange equations. In other words, it is a Lie symmetry of these equations:
\begin{equation}
\op X^{[2]}(\op E_i(L))=0\quad{\rm when}\quad\{\op E_j(L)=0\},\label{Lie}
\end{equation}
where we introduced the second prolongation of $\op X$:
\begin{equation*}
\op X^{[2]}:=\op X^{[1]}+(\ddot\xi^{\,i}-\dot q^i\ddot\tau-2\ddot q^{\,i}\dot\tau)\frac{\partial}{\partial\ddot q^{\,i}}\,.
\end{equation*}
The converse is not true, in general. For instance, $\partial_t$ is basically a Lie symmetry of $\boldsymbol{\op E}(L_{\rm BCK})=0$ whatever $V(q)$ be but it is never a Noether symmetry of $L_{\rm BCK}$. Thus, one expects more symmetries by seeking solutions of~\eref{Lie} and, by the way, more first integrals.

Yet again, let us restrict ourself to Lie point symmetries (LPS) in one dimension. Each of them `carries' naturally a first integral which may be extracted as follows. We first determine a zero order $x=x(q,t)$ and a first order $y=y(q,\dot q,t)$ invariants of $\op X^{[1]}$. Then, according to Lie's theory, the second order differential equation $\op E(L)=0$ may be reduced to a first order one, say $\Phi(x,y,y'_x)=0$, with
\begin{equation*}
y'_x:=\frac{\op D(y)}{\op D(x)}\,.
\end{equation*}
Its general solution admits an implicit form $G(x,y)=C$, where $C$ is an integration constant; whence,
\begin{equation*}
I(q,\dot q,t):=G(x(q,t),y(q,\dot q,t))=\mathrm{cst}.\quad{\rm when}\quad\op E(L)=0.
\end{equation*}
Since $G$ is a function of $\op X$'s invariants, $I$ is also an invariant of the LPS, i.e. $\op X^{[1]}(I)=0$. Applying the method to the LPS $\partial_t$ of the Euler-Lagrange equation derived from $L_{\rm BCK}$, with $x=q$ and $y=\dot q$, one obtains an Abel equation of the second kind, viz. 
\begin{equation*}
yy'_x+2\gamma y+\partial_xV(x)=0.
\end{equation*}
It must be solved to deduce the first integral, which has the great advantage of being explicitly time-independent. The other side of the coin is that it is generally a hard task which necessitates tedious algebraic and analytic manipulations to obtain cumbersome solutions in closed or parametric forms~\cite{Polyanin}. However, the quadratic~\cite{Denman1968} and linear potentials are notable exceptions (see $I_6$ in table~\ref{tableLPS}). 

Let us now move on to the search for all the LPS of $\op E(L_{\rm BCK})$. The determining equation~\eref{Lie} reads in this context
\begin{equation*}
\ddot\xi-\dot q\ddot\tau-2\ddot q\dot\tau+2\gamma(\dot\xi-\dot q\dot\tau)+\xi V''=0\quad{\rm with}\quad\ddot q=-V'-2\gamma\dot q.
\end{equation*}
Expanding the time-derivatives, one recovers a polynomial identity of the third degree. The equations relatives to $\dot q^3$ and $\dot q^2$ restrict as usually the form of $\tau$ and $\xi$:
\begin{eqnarray*}
\tau=a(t)+b(t)q\,,\\
\xi=(\dot a-2\gamma a)q^2+c(t)q+d(t)\,,
\end{eqnarray*}
whereas the two leftover equations produce compatibility equations:
\begin{eqnarray}
3(\ddot a-2\gamma\dot a)q-\ddot b+2\gamma\dot b+2\dot c+3aV'=0\;;\label{aV'}\\
((\dot a-2\gamma a)q^2+cq+d)V''+(4\gamma aq+2\dot b-c)V'\nonumber\\
\qquad\qquad\qquad\qquad\qquad\;+(\dddot a-4\gamma^2\dot a)q^2+(\ddot c+2\gamma\dot c)q+\ddot d+2\gamma\dot d=0.
\end{eqnarray}
Equation~\eref{aV'} imposes $a=0$ unless $V$ is linear or quadratic: here again, one observes the dichotomy between these peculiar potentials and all the others. 

The full analysis of Lie point symmetries in the linear and quadratic cases is straightforward and leads to eight independent LPS, including the five NPS previously identified. Apart from $\partial_t$, the new symmetries have a $q$-dependent $\tau$ function. One can find their expressions e.g. in reference~\cite{Pandey}, and the complete Lie group analysis of the quadratic case, including the associated first integrals, in the work of Cerver\'o and Villarroel~\cite{Cervero}. For the sake of completeness, we summarize the three new symmetries of the linear potential in table~\ref{tableLPS} together with their first integrals. Moreover, according to the converse of Noether's theorem and formula~\eref{Hij}, any of those three first integrals $I_k$ is generated by a class of Noether symmetries $\{Y_k\}$ whose characteristics equal
\begin{equation*}
\mu_k=\rme^{-2\gamma t}\,\frac{\partial I_k}{\partial\dot q}\,.
\end{equation*}
We gave in table~\ref{tableLPS} their evolutionary representative $\mu_k\partial_q$\footnote{We could have chosen the \textit{contact representative} which is more satisfying in both mathematical and physical senses (contact transformations play central roles in wave theory, especially at the level of Huygens' principle~\cite{IbragimovReview}, and in Hamiltonian mechanics~\cite{Whittaker}).}.

\begin{table}\small
\centering
\renewcommand{\arraystretch}{1.5}
\begin{tabular}{llll}
\br
Lie point symmetry & First integral & Noether symmetry\\
\mr
$\op X_6=\partial_t$ & $I_6=F\log|I_1|-2\gamma I_2$ & $\op Y_6=\frac{\dot q}{I_1}\,\partial_q$\\
$\op X_7=(Ft-2\gamma q)\partial_t+(Ft-2\gamma q)^2\partial_q$ & $I_7=\frac{I_1}{2\gamma I_2+F}$ & $\op Y_7=\frac{Ft-2\gamma q}{(2\gamma I_2-F)^2}\,\partial_q$\\
$\op X_8=\rme^{2\gamma t}(Ft-2\gamma q)(2\gamma\partial_t+F\partial_q)$ & $I_8=\frac{2\gamma I_2+F}{I_1}$ &  $\op Y_8=\frac{Ft-2\gamma q}{I_1^2}\,\partial_q$\\
\br
\end{tabular}
\caption{The three additional Lie point symmetries for the linear potential $V=-Fq$, their associated first integrals, and the Noether symmetry generating the same conservation law (in the evolutionary form). See table~\ref{tableNPSlineaire} for the expression of the first integrals $I_1$ and $I_2$.}\label{tableLPS}
\end{table}

Let us now consider the other potentials. The condition $a=0$ restricts once again $\tau$ as a pure function of $t$. Then, equation~\eref{aV'} gives $c$ as a function of $\tau=b$,
\begin{equation*}
c=\frac12(\dot\tau-2\gamma\tau)+K\quad(K=\rm{cst.})
\end{equation*} 
and the last compatibility equation becomes
\begin{equation*}
\fl\left[\frac12\left(\dot\tau-2\gamma\tau+2K\right)q+d\right]V''+\frac12\left(3\dot\tau+2\gamma \tau-2K\right)V'+\frac12\left(\ddot\tau-4\gamma^2\dot\tau\right)q+\ddot d+2\gamma\dot d=0.
\end{equation*}
Integrating this equation with respect to $q$, we get at the right-hand side a function of $t$, say $g(t)$. The resulting expression, in its entirety, is precisely the compatibility equation~\eref{EdC} peculiar to the NPS, where $d=\psi$, $g=-\dot\chi\rme^{2\gamma t}$ and once the replacement $\tau\to\tau-K/\gamma$ is done. Hence, apart from $\partial_t$ and the NPS already found, there is no new Lie point symmetries.

\section{Weak Noether's theorem and the action of Bateman-Caldirola-Kanai}\label{sec:action}

As pointed out in \sref{sec:Noetherthm}, the existence of integrating factors is a key-feature of first integrals. However, one could imagine some arbitrary combination $\nu_i(q,\dot q,t)\op E_i(L)$ which, by construction, is vanishing along the solutions of the Euler-Lagrange equations. Generically, the functions $\nu_i$ are not integrating factors of these equations, so that the quantity
\begin{equation*}
\mathcal I=\int\nu_i(q,\dot q,t)\op E_i(L)\,\rmd t,
\end{equation*} 
albeit constant along the actual paths, is not a first integral. Indeed, it is not expressible as a function of $(q,\dot q,t)$ since it depends on the whole history of the motion. One says that $\mathcal I$ is a \emph{non-local constant of motion}. We are in a position to state a weakened version of Noether's theorem.

\begin{theorem}[Weak Noether's theorem]
For any transformation~\eref{transformation}, the quantity
\begin{equation}
\mathcal I=\int(\xi_i-\dot q_i\tau)\op E_i(L)\,\rmd t=\int\big(\op X^{[1]}(L)+\dot\tau L\big)\,\rmd t-L\tau-\frac{\partial L}{\partial\dot q_i}(\xi_i-\dot q_i\tau)\label{Cstofmotion}
\end{equation}
is a constant of motion of the problem.
\end{theorem}

Such a conservation law seems unwieldy unless we manage to express the integral in~\eref{Cstofmotion} as a function of $(q,\dot q,t)$ along the solution curves, that is, to prove the existence of a function $f(q,\dot q,t)$ such that
\begin{equation*}
\int(\op X^{[1]}(L)+\dot\tau L)\,\rmd t=f(q,\dot q,t)\quad{\rm when}\quad \{\op E_i(L)=0\}.
\end{equation*} 
Of course, this condition is fulfilled if we set
\begin{equation}
f(q,\dot q,t)=L\tau+\frac{\partial L}{\partial\dot q_i}(\xi_i-\dot q_i\tau)+{\rm cst.}\label{truism}
\end{equation}
but such a naive choice leads to a truism: scalar numbers are constants of motion. Actually, any alternative expression $f(q,\dot q,t)$ provides a true first integral (this is certainly the case if $f$ does not depend on $\dot q$). It amounts to ask for a function $f$, excluding~\eref{truism}, such that
\begin{equation*}
\op X^{[1]}(L)+\dot\tau L=\dot f\quad{\rm when}\quad\{\op E_i(L)=0\}.
\end{equation*}
Put another way, the transformation $\op X$ must be an `on-flow' solution of the Rund-Trautman identity. We mention that, recently, Gorni and Zampieri~\cite{Gorni_Zampieri,Gorni_Zampieri2} used this weakened Noether's theorem to derive alternative constructions of some known first integrals. 

The theorem gives us access to non-local extensions of the usual conservation laws. For instance, the time translation $\partial_t$ is a Noether symmetry if and only if, according to~\eref{integratingfactorNS}, $\mu^i=\dot q^i$ are integrating factors. When $L$ is time-independent, this is certainly the case and the corresponding first integral is the Hamiltonian $H=\dot q^i\partial_{\dot q^i}L-L$. Generally, it comes down to the law of conservation of mechanical energy. On the contrary, if $\partial_t$ is not a Noether symmetry one can nonetheless associate to $\partial_t$ the constant of motion
\begin{equation*}
\mathcal I=-\int\dot q_i\,\op E_i(L)\,\rmd t=H+\int\frac{\partial L}{\partial t}\,\rmd t.
\end{equation*}
It reflects basically the relation $\partial_tL=-\dot H$ along the solution curves. Returning to $L_{\rm BCK}$, this constant of motion acquires an interesting physical meaning. Indeed, noticing the singular property
\begin{equation*}
\frac{\partial L_{\rm BCK}}{\partial t}=2\gamma\,L_{\rm BCK},
\end{equation*}
one obtains the kinematic conservation of
\begin{equation}
\mathcal I=H+2\gamma A_{\rm BCK},\label{cstaction}
\end{equation}
where
\begin{equation*}
A_{\rm BCK}=\int L_{\rm BCK}\,\rmd t
\end{equation*}
is the action based on $L_{\rm BCK}$. Consequently, along the solution curves, the action is expressible as a function of $(\boldsymbol r,\boldsymbol{\dot r},t)$:
\begin{equation}
A_{\rm BCK}=\frac{\mathcal I-H}{2\gamma}\quad{\rm when}\quad\boldsymbol{\op E}(L_{\rm BCK})=\boldsymbol 0,\label{cstaction2}
\end{equation} 
where $\mathcal I$ plays merely the role of an integrating constant. Now, suppose that one manages to find an alternative expression of $A_{\rm BCK}$, i.e. a function  $G(\boldsymbol r,\boldsymbol{\dot r},t)$ such that
\begin{equation*}
A_{\rm BCK}=G(\boldsymbol r,\boldsymbol{\dot r},t)\not\equiv-\frac{H}{2\gamma}+{\rm cst.}\quad{\rm when}\quad\boldsymbol{\op E}(L_{\rm BCK})=\boldsymbol 0.
\end{equation*} 
Then, by injecting this expression in~\eref{cstaction}, $\mathcal I$ becomes a non-trivial first integral $I$. Conversely, if $\mathcal I$ is so, one extracts from~\eref{cstaction2} an expression of $A_{\rm BCK}$ along the solution curves. To sum up, an alternative local expression of the action along the actual paths amounts to a first integral, i.e. to a Noether symmetry. According to~\eref{FINoether}, this relationship is given by
\begin{equation*}
A_{\rm BCK}=\frac{1}{2\gamma}\left[f-(1+\tau)\left(\frac12\,\boldsymbol{\dot r}^2+V\right)\rme^{2\gamma t}-\boldsymbol\xi\cdot\boldsymbol{\dot r}\,\rme^{2\gamma t}\right],
\end{equation*}
where $\tau$ is chosen dimensionless.

\section{Conclusion}

In this article, we addressed the question of the conservation laws in the classical dynamics of a particle submitted to a linear dissipation, basing our study on the variational framework offered by the celebrated BCK Lagrangian. We found all the time-independent potentials $V(q)$ for which the Lagrangian admits Noether point symmetries in the unidimensional case. Then, we showed that not all of them `resist' under the direct mapping $V(q)\to V(r)$ to the central case in three dimensions. As a next stage, we established a systematic method to transform Lagrangians enjoying a variational point symmetry into autonomous others, and applied it in our context. For the sake of completeness, we also investigated the Lie point symmetries which form an overalgebra of Noether ones. Apart from the family of potentials at most quadratic, whose point symmetry algebra gains three dimensions (from 5 to 8), the others acquire only the obvious time translation symmetry. Finally, we enlightened the interconnection between symmetries and the BCK action functional, in relation with a weakened version of Noether's theorem and the idea of on-flow solutions of the Rund-Trautman identity.

\ack

We wish to thank all the colleagues of the \textit{Groupe de Physique Statistique} for many enlightening discussions. One of us (R. L.) is grateful for their kind hospitality. 

\appendix

\section*{Appendix}

\setcounter{section}{1}

We give here a simple proof of formula~\eref{appendice}. Using the hypothesis and notations of \sref{sec:autonomous}, the equality $\widetilde L=L\opt D(t)$, i.e. $L\rmd t=\widetilde L\rmd T$, is tantamount to the invariance of the Poincar\'e-Cartan form, that is
\begin{equation*}
p_i\rmd q^i-H\rmd t=P_i\rmd Q^i-\widetilde H\rmd T.
\end{equation*}
Setting $q^0=t$, $p_0=-H$, $Q^0=T$, $P_0=-\widetilde H$, this invariance reads
\begin{equation}
p_\mu\rmd q^\mu=P_\mu\rmd Q^\mu\,,\label{Cartan}
\end{equation}
where the index $\mu$ goes from 0 to $n$. Furthermore, the first integrals $I$ and $\widetilde I$ induced by the common point symmetry $\op X$ of $L$ and $\widetilde L$ become simply
\begin{equation*}
I=f-p_\mu\,\op X(q^\mu)\quad{\rm and}\quad\widetilde I=f-P_\mu\,\op X(Q^\mu).
\end{equation*}
The equality~\eref{Cartan} allows to determine the new momenta from the old ones through
\begin{equation}
p_\nu\frac{\partial q^\nu}{\partial Q^\mu}\rmd Q^\mu=P_\mu\rmd Q^\mu\quad\Longrightarrow\quad P_\mu=p_\nu\,\frac{\partial q^\nu}{\partial Q^\mu}\,.
\end{equation}
Then,
\begin{equation*}
\widetilde I=f-p_\nu\,\frac{\partial q^\nu}{\partial Q^\mu}\op X(Q^\mu)=f-p_\nu\op X(q^\nu)=I.
\end{equation*}
This result is nothing else but the scalar nature of $p_\mu\,\op X(q^\mu)$ inherited from the covariance of the momenta $p_\mu$ and the contravariance of the coordinates $q^\mu$.

\bibliography{Leone_Gourieux}

\end{document}